\documentclass[lettersize,journal]{IEEEtran}
\usepackage{amsmath,amsfonts}
\usepackage{algorithmic}
\usepackage{array}
\usepackage[caption=false,font=normalsize,labelfont=sf,textfont=sf]{subfig}
\usepackage{textcomp}
\usepackage{stfloats}
\usepackage{url}
\usepackage{verbatim}
\usepackage{graphicx}
\usepackage{enumitem}
\usepackage{amsthm}
\newtheorem{definition}{Definition}
\newtheorem{theorem}{Theorem}
\newtheorem{lemma}{Lemma}
\newtheorem{remark}{Remark}
\newtheorem{corollary}{Corollary}

\def\BibTeX{{\rm B\kern-.05em{\sc i\kern-.025em b}\kern-.08em
    T\kern-.1667em\lower.7ex\hbox{E}\kern-.125emX}}
\usepackage{balance}

\begin{document}
\title{Cryptanalysis of a Class of Ideal Secret Sharing Schemes Based on the CRT for Polynomial Rings}
\author{Jian~Ding,~Zhiyuan Guo,~Hongju Li, and~Yuying Man
\thanks{This work was supported in part by the National Natural Science Foundation of China (NSFC) Young Scientists Fund (Category C) under Grant No. 126011024, in part by Natural Science Research Project of Anhui Higher Education Institutions under Grants No. 2024AH051324 and No. 2025AHGXZK30056, and in part by Projects of Chaohu University under Grants No. KYQD-202220, No. 2024cxtd147 and No. hxkt20250327.}
\thanks{J. Ding, H. Li and Y. Man are with School of Mathematics and Big Data, Chaohu University, Hefei, China
(e-mail: dingjian\_happy@163.com, 058072@chu.edu.cn, yuying.man@aliyun.com).}
\thanks{Z. Guo is with School of Mathematics and Statistics, Hefei University, Hefei, China
(e-mail: gzy\_26105@163.com.)}}


\maketitle

\begin{abstract}
  Based on the CRT for polynomial rings, Yang, Zhu, Fu and Xia (ISIT 2026) proposed a compartmented secret sharing scheme for the compartmented access structure with lower bounds, claiming it can be ideal. We exhibit three families of unauthorized subsets that reconstruct the secret. The scheme is therefore insecure, except for a degenerate choice of parameters in which the only authorized subset is the whole participant set. The flaw in the security analysis is that it does not take all published polynomials into account. Similar flaws appear in the same first author's hierarchical schemes (ISIT 2024, 2026).
\end{abstract}


\section{Introduction}
In a secret sharing scheme (SSS) \cite{Shamir1979, Blakley1979}, a dealer splits a secret into shares and distributes them among a set $\mathcal{P}$ of $n$ participants, in such a way that the authorized subsets of $\mathcal{P}$ can reconstruct the secret while the unauthorized ones cannot. The collection of all authorized subsets forms the \emph{access structure}. An SSS is \emph{perfect} if every unauthorized subset obtains no information about the secret, and \emph{ideal} if it is perfect and has information rate one.

The first proposed secret sharing schemes \cite{Shamir1979,Blakley1979} have the \emph{$(t,n)$-threshold} access structure:
\[\Gamma=\{\mathcal{A}\subseteq \mathcal{P}: |\mathcal{A}|\geq t\}.\]
In such a scheme, all the participants play exactly the same role, and none of them is more privileged than another. Many applications, however, involve participants with different degree of trust. Compartmented access structures, introduced by Simmons \cite{Simmons1990}, are suitable for this scenario. In a \emph{compartmented} secret sharing scheme (CSSS), the participant set $\mathcal{P}$ is partitioned into compartments $\gamma_1,\gamma_2,\ldots,\gamma_q$, each equipped with a threshold $t_i$, and a global threshold $t$ is imposed on $\mathcal{P}$; the resulting \emph{compartmented access structure with lower bounds} (LCAS) is
       \[\Gamma_L=\{\mathcal{A} \subseteq \mathcal{P}: \forall 1\leq i\leq q,\ |\mathcal{A}\cap \gamma_i|\geq t_i~\mathrm{and}~|\mathcal{A}|\geq t\}.\]

Many ideal CSSSs for LCASs have been constructed from linear algebra \cite{Brickell1989}, bivariate interpolation \cite{TassaDyn2009}, polymatroid theory \cite{Farras2012,Chenqi2019}, and coding theory \cite{Chen2022}. However, some of these constructions have limitations. The construction of \cite{Brickell1989} requires checking the non-singularity of an exponential number of matrices. That of \cite{TassaDyn2009} is probabilistic and only statistically perfect. In contrast, the CRT-based construction over the integers \cite{TipleaDragan2018} is non-ideal. Recently, Yang, Zhu, Fu and Xia \cite{YangCSS2026} proposed an explicit CSSS for LCAS based on the CRT for polynomial rings. They claimed that it is perfect and can be made ideal by an appropriate choice of the moduli.

\emph{Our contributions.} We point out that the scheme of \cite{YangCSS2026} is insecure, and analyse the root cause of its insecurity. More precisely, we give the following contributions.
\begin{enumerate}
  \item We exhibit three families of unauthorized subsets that reconstruct the secret (Theorem \ref{Th: attack}). The scheme is therefore insecure, for every choice of the parameters except the degenerate one in which the only authorized subset is $\mathcal{P}$ itself (Corollary \ref{Cor: witnesses}).

  \item We find that the flaw in the security analysis of \cite{YangCSS2026} is its failure to take all the published polynomials into account. Similar flaws also appear in the hierarchical secret sharing schemes by the same first author \cite{Yangjing2024,YangjingHSS2026}, whose security analysis omits the published polynomials in exactly the same way.
\end{enumerate}

The rest of the paper is organised as follows. Section \ref{Sec: prelim} introduces the basic notions of secret sharing and the CRT for polynomial rings. Section \ref{Sec: review} reviews the scheme of \cite{YangCSS2026}. Section \ref{Sec: attack} presents the cryptanalysis, together with the root cause of the insecurity, and Section \ref{Sec: conclusion} concludes the paper.

\section{Preliminaries}\label{Sec: prelim}
\begin{definition}[Secret sharing scheme]
Let $\mathcal{P}=\{P_1,P_2,\ldots,\\P_n\}$ denote the set of all participants. The secret space is denoted by $\mathcal{S}$, while $\mathcal{S}_\theta$ stands for the share space of the $\theta$-th participant. A secret sharing scheme of $n$ participants generally comprises two procedures: share distribution and secret reconstruction.
	\begin{enumerate}
	  \item Share Distribution Phase: For any secret drawn from $\mathcal{S}$ and a certain random input selected from $\mathcal{R}$, the dealer invokes the share generation algorithm
		\[\mathsf{SHAR}\colon \mathcal{S}\times \mathcal{R}\to \mathcal{S}_1\times \mathcal{S}_2\times\dotsb\times \mathcal{S}_n\]
		to allocate shares among the participants in $\mathcal{P}$.
	  \item Secret Reconstruction Phase: Any authorized subset $\mathcal{A}\subseteq \mathcal{P}$ is capable of reconstructing the secret through the secret reconstruction algorithm
		\[\mathsf{RECO}\colon\prod_{P_\theta\in \mathcal{A}} \mathcal{S}_\theta\to \mathcal{S},\]
		whereas any unauthorized subset is not.
	\end{enumerate}

\end{definition}

All authorized subsets form a set called the access structure. In the sequel, each participant $P_\theta$ is frequently identified with its index $\theta$. For a positive integer $n$ we write $[n]=\{1,2,\ldots,n\}$, and so that $\mathcal{P}=[n]$.

\begin{definition}[Compartmented Secret Sharing Scheme, \cite{Chenqi2019}]\label{def: CSS}
Let $\mathcal{P}=\{P_1,P_2,\ldots,P_n\}$ be the set of all participants, partitioned by the dealer into $q$ disjoint compartments $\gamma_1, \gamma_2, \ldots, \gamma_q$, i.e.,
                       \[\mathcal{P}=\bigcup_{i=1}^{q} \gamma_i,~\mathrm{and}~\gamma_i\cap\gamma_j=\emptyset~\mathrm{for~all}~1\leq i<j\leq q.\]
Denote by $|\cdot|$ the cardinality of a set. Let $t_1, t_2, \ldots, t_q,t$ be positive integers such that
                      \[t_i\leq |\gamma_i|=n_i (1\le i\le q),~\mathrm{and}~\sum_{i=1}^q t_i\leq t\leq n.\]
The compartmented access structure with lower bounds (LCAS) is
                 \[\Gamma_L=\{\mathcal{A} \subseteq \mathcal{P}: \forall 1\leq i\leq q,\ |\mathcal{A}\cap \gamma_i|\geq t_i~\mathrm{and}~|\mathcal{A}|\geq t\}.\]
A Compartmented Secret Sharing Scheme (CSSS) for the LCAS with perfect security is a secret sharing scheme satisfying the following correctness and privacy requirements.
\begin{enumerate}
\item Correctness. The secret can be reconstructed from any authorized subset $\mathcal{A}\in\Gamma_L$.
\item Privacy. No information about the secret can be obtained from any subset $\mathcal{B}\notin \Gamma_L$.
\end{enumerate}
\end{definition}

Throughout this paper we assume $q\geq 2$, the case $q=1$ reducing to a threshold access structure.

\begin{lemma}[CRT for polynomial rings, \cite{YangCSS2026}]\label{le: CRT}
Let $p$ be a prime, $\mathbb{F}_p$ the finite field with $p$ elements, and $\mathbb{F}_p[x]$ the univariate polynomial ring over $\mathbb{F}_p$. Let $m_1(x),m_2(x),\ldots,m_k(x)\in \mathbb{F}_p[x]$ be pairwise coprime, let $M(x)=\prod_{i=1}^{k}m_i(x)$, $M_i(x)=M(x)/m_i(x)$, and $\lambda_i(x)\equiv M_i^{-1}(x)\pmod{m_i(x)}$. Then, for any $f_1(x),\ldots,f_k(x)\in \mathbb{F}_p[x]$, the system of congruences $f(x)\equiv f_i(x)\pmod{m_i(x)}, i\in[k]$, has the solution
               \[f(x)\equiv \sum_{i=1}^{k}\lambda_i(x)M_i(x)f_i(x)\pmod{M(x)}.\]
If $\deg(f(x))<\deg(M(x))$, then the solution is unique, and we denote it by
               \[f(x)=\sum_{i=1}^{k}\lambda_i(x)M_i(x)f_i(x)\bmod{M(x)}.\]
\end{lemma}

\section{Review of the scheme of Yang et al.}\label{Sec: review}

 Let $\mathcal{P}=\{P_1,P_2,\ldots,P_n\}$ be the set of all participants and let it be partitioned into $q\geq 2$ disjoint compartments $\gamma_1,\gamma_2,\ldots,\gamma_q$, where $|\gamma_i|=n_i$ for $i\in [q]$. Set $N_j=\sum_{i=1}^{j}n_i$ and $N_0=0$, so that $N_q=n$. The compartment thresholds satisfy $0<t_i\leq n_i$ for all $i\in [q]$, and the global threshold satisfies $\sum_{i=1}^{q}t_i\leq t\leq n$. The access structure is
             \[\Gamma_L=\{\mathcal{A} \subseteq \mathcal{P}: \forall 1\leq i\leq q,\ |\mathcal{A}\cap \gamma_i|\geq t_i~\mathrm{and}~|\mathcal{A}|\geq t\}.\]
For $u\in [q]$ and $v\in [n_u]$, the $v$-th participant of $\gamma_u$ is denoted by $P_{u,v}$, which coincides with $P_{\theta}$ for $\theta=N_{u-1}+v$. The scheme of \cite{YangCSS2026} consists of the share distribution phase and the secret reconstruction phase.

\textbf{1) Share Distribution Phase.}
\begin{enumerate}[label=Step \arabic*., leftmargin=*]
   \item The dealer $D$ chooses a prime $p$ and an integer $d_0\geq 1$, and sets $m_0(x)=x^{d_0}\in \mathbb{F}_p[x]$.
   \item $D$ chooses pairwise coprime polynomials $m_{\theta}(x)\in \mathbb{F}_p[x]$, $\theta\in [n]$, and assigns $m_{\theta}(x)=m_{u,v}(x)$ to $P_{\theta}=P_{u,v}$ with $\theta=N_{u-1}+v$. Denote by $d_{\theta}=\deg(m_{\theta}(x))$ and $d_{u,v}=\deg(m_{u,v}(x))$, then these polynomials satisfy
   \begin{enumerate}[label=(\roman*)]
      \item $m_0(x)$ and $m_{\theta}(x)$ are coprime for all $\theta\in [n]$;
      \item $d_0\leq d_1\leq d_2\leq \cdots \leq d_n$, that is, $d_0\leq d_{1,1}\leq \cdots \leq d_{1,n_1}\leq d_{2,1}\leq \cdots \leq d_{q,n_q}$;
      \item $d_0+\sum\limits_{\theta^{(i)}=n-t_i+2}^{n}d_{\theta^{(i)}}\leq \sum\limits_{\theta^{(i)}=1}^{t_i}d_{\theta^{(i)}}$ for all $i\in [q]$;
      \item $d_0+\sum\limits_{\theta=n-t+2}^{n}d_{\theta}\leq \sum\limits_{\theta=1}^{t}d_{\theta}$.
   \end{enumerate}
   \item $D$ publishes $m_0(x)$ and $m_{\theta}(x)$ for $\theta\in [n]$.
   \item $D$ chooses the secret $s(x)$ uniformly at random from $\mathcal{S}=\{g(x)\in \mathbb{F}_p[x]: \deg(g(x))<d_0\}$.
   \item $D$ chooses $q$ piece polynomials $s_1(x),\ldots,s_q(x)$ independently and uniformly from $\mathcal{S}$, and sets $s_{q+1}(x)=(s(x)-s_1(x)-\cdots-s_q(x))\bmod{m_0(x)}$.
   \item For each $u\in [q]$, $D$ chooses $\alpha'_u(x)$ uniformly from $\mathcal{G}'_u=\{g(x)\in \mathbb{F}_p[x]: \deg(g(x))<(\sum_{\theta^{(u)}=1}^{t_u}d_{\theta^{(u)}})-d_0\}$ and forms the compartment secret polynomial
        \[s'_u(x)=s_u(x)+\alpha'_u(x)x^{d_0}.\]
   \item $D$ chooses $\alpha'_{q+1}(x)$ uniformly from $\mathcal{G}'_{q+1}=\{g(x)\in \mathbb{F}_p[x]: \deg(g(x))<(\sum_{\theta=1}^{t}d_{\theta})-d_0\}$ and forms the global secret polynomial
        \[s'_{q+1}(x)=s_{q+1}(x)+\alpha'_{q+1}(x)x^{d_0}.\]
   \item Each participant $P_{\theta}=P_{u,v}$ receives the share
        \[s_{u,v}(x)=s'_u(x)\bmod{m_{u,v}(x)},\]
        and $D$ publishes
        \[\omega_{u,v}(x)=(s'_{q+1}(x)-s_{u,v}(x))\bmod{m_{u,v}(x)}.\]
\end{enumerate}

\textbf{2) Secret Reconstruction Phase.} Let $\mathcal{A}\in \Gamma_L$ and $\mathcal{A}^{(i)}=\mathcal{A}\cap \gamma_i$. For each $i\in [q]$ the participants of $\mathcal{A}^{(i)}$ pool their private shares and reconstruct
        \[s'_i(x) = \sum_{{\theta^{(i)}}\in \mathcal{A}^{(i)}} \lambda_{\theta^{(i)}}(x) M_{\theta^{(i)}}(x) c_{\theta^{(i)}}(x) \bmod{M_{\mathcal{A}^{(i)}}(x)},\]
where $c_{\theta^{(i)}}(x)=s_{u,v}(x), m_{\theta^{(i)}}(x)=m_{u,v}(x)$ for $\theta^{(i)}=N_{u-1}+v$, and $M_{\mathcal{A}^{(i)}}(x) = \prod_{{\theta^{(i)}} \in \mathcal{A}^{(i)}} m_{\theta^{(i)}}(x)$, $M_{\theta^{(i)}}(x) = M_{\mathcal{A}^{(i)}}(x)/m_{\theta^{(i)}}(x), \lambda_{\theta^{(i)}}(x) \equiv M_{\theta^{(i)}}^{-1}(x)\ (\mathrm{mod}\ m_{\theta^{(i)}}(x))$.

Next, the participants of $\mathcal{A}$ pool their private shares and public shares to reconstruct
\[s'_{q+1}(x) = \sum_{{\theta} \in \mathcal{A}} \lambda_{\theta}(x) M_{\theta}(x) c_{\theta}(x) \bmod{M_\mathcal{A}(x)},\]
where $c_{\theta}(x)=\omega_{u,v}+s_{u,v}(x), m_{\theta}(x)=m_{u,v}(x)$ for $\theta=N_{u-1}+v$, and $M_{\mathcal{A}}(x) = \prod_{{\theta} \in \mathcal{A}} m_{\theta}(x)$, $M_{\theta}(x) = M_{\mathcal{A}}(x)/m_{\theta}(x), \lambda_{\theta}(x) \equiv M_{\theta}^{-1}(x)\ (\mathrm{mod}\ m_{\theta}(x))$.

Finally, $s_i(x) = s'_i(x) \bmod{m_0(x)}$ for $i\in [q+1]$, and $s(x) = (s_1(x) + \cdots + s_q(x) + s_{q+1}(x))\bmod{m_0(x)}$.

\section{Cryptanalysis}\label{Sec: attack}

Subsection~\ref{Subsec:attack} points out that the scheme of \cite{YangCSS2026} has a security flaw: some unauthorized subsets are able to reconstruct the secret. Subsection~\ref{Subsec: discussion} discusses the flaw in the security analysis of \cite{YangCSS2026}.

\subsection{Unauthorized subsets that reconstruct the secret}\label{Subsec:attack}

For $\mathcal{X}\subseteq [n]$ we write $M_{\mathcal{X}}(x)=\prod\limits_{\theta\in \mathcal{X}}m_{\theta}(x)$ and $D_{\mathcal{X}}=\sum\limits_{\theta\in \mathcal{X}}d_{\theta}=\deg(M_{\mathcal{X}}(x))$. Denote by
\[T_u=\sum\limits_{\theta=1}^{t_u}d_{\theta}~(u\in [q]), ~\mathrm{and}~ T=\sum\limits_{\theta=1}^{t}d_{\theta}.\]
For $\mathcal{B}\subseteq [n]$, we let $\mathcal{B}\setminus\gamma_u=\{\theta:\theta\in \mathcal{B}~\mathrm{and}~\theta\notin \gamma_u\}$. Next, we present basic properties of the scheme of \cite{YangCSS2026} in Lemmas \ref{lemma:1} and \ref{le: leakage}. Based on these properties, we give attacks on their scheme in Theorem \ref{Th: attack}. We then show in Corollary \ref{Cor: witnesses} when these attacks apply to unauthorized subsets.

\begin{lemma}\label{lemma:1}
Let $\mathcal{X}\subseteq [n]$, $u\in [q]$ and let $k$ be a positive integer. Then the following properties hold in the scheme of \cite{YangCSS2026}.
\begin{enumerate}[label=(\roman*)]
   \item $\deg(s'_u(x))<T_u$ and $\deg(s'_{q+1}(x))<T$.
   \item $D_{\mathcal{X}}\geq \sum_{\theta=1}^{k}d_{\theta}$ if $|\mathcal{X}|\geq k$.
   \item $T\geq T_u$.
   \item $D_{\gamma_u}\geq T_u$.
   \item $s(x)=\big(\sum_{u=1}^{q}s'_u(x)+s'_{q+1}(x)\big)\bmod{m_0(x)}$.
\end{enumerate}
\end{lemma}

\begin{proof} (i) Recall that $s'_u(x)=s_u(x)+\alpha'_u(x)x^{d_0} (u\in [q])$ and $s'_{q+1}(x)=s_{q+1}(x)+\alpha'_{q+1}(x)x^{d_0}$,
   where $s_u(x)\in \mathcal{S}=\{g(x)\in \mathbb{F}_p[x]: \deg(g(x))<d_0\}$, $\alpha'_u(x)\in \mathcal{G}'_u=\{g(x)\in \mathbb{F}_p[x]: \deg(g(x))<(\sum_{\theta^{(u)}=1}^{t_u}d_{\theta^{(u)}})-d_0=T_u-d_0\}$ for $u\in [q]$, $s_{q+1}(x)=(s(x)-s_1(x)-\cdots-s_q(x))\bmod{m_0(x)}$, and $\alpha'_{q+1}(x)\in\mathcal{G}'_{q+1}=\{g(x)\in \mathbb{F}_p[x]: \deg(g(x))<(\sum_{\theta=1}^{t}d_{\theta})-d_0=T-d_0\}$. Consequently, $\deg(s'_u(x))\leq\deg(\alpha^{\prime}_u(x))+d_0<\sum_{\theta^{(u)}=1}^{t_u}d_{\theta^{(u)}}=T_u$ for $u\in [q]$, and $\deg(s'_{q+1}(x))\leq\deg(\alpha^{\prime}_{q+1}(x))+d_0<\sum_{\theta=1}^{t}d_{\theta}=T$.

(ii) Since $d_1\leq d_2\leq \cdots \leq d_n$, every $\mathcal{X}\subseteq [n]$ satisfies $D_{\mathcal{X}}\geq \sum_{\theta=1}^{|\mathcal{X}|}d_{\theta}$. Therefore, $D_{\mathcal{X}}\geq \sum_{\theta=1}^{k}d_{\theta}$ if $|\mathcal{X}|\geq k$.

(iii) Since $t\geq \sum_{i=1}^{q}t_i\geq t_u$ for $u\in [q]$, then
                \[T=\sum\limits_{\theta=1}^{t}d_{\theta}\geq \sum\limits_{\theta=1}^{t_u}d_{\theta}=T_u.\]
(iv)Since $|\gamma_u|=n_u\geq t_u$, then $D_{\gamma_u}\geq T_u$ by the above item (ii).

(v) Recall that $m_0(x)=x^{d_0}$, $s_u(x)=s'_u(x)\bmod{m_0(x)}$ and $s_{q+1}(x)=\big(s(x)-\sum_{u=1}^{q}s_u(x)\big)\bmod{m_0(x)}$. Therefore,
       \begin{displaymath}
         \begin{aligned}
         s(x)=&\left(\sum_{u=1}^{q}s_u(x)+s_{q+1}(x)\right)\bmod{m_0(x)}\\
           =&\left(\sum_{u=1}^{q}s^{\prime}_u(x)+s^{\prime}_{q+1}(x)\right)\bmod{m_0(x)}.
       \end{aligned}
       \end{displaymath}
\end{proof}

\begin{lemma}\label{le: leakage}
Denote by $g_u(x)=s'_{q+1}(x)-s'_u(x)$ for $u\in [q]$, then the following properties hold in the scheme of \cite{YangCSS2026}.
\begin{enumerate}[label=(\roman*)]
   \item $\deg(g_u(x))<T$ and $g_u(x)\equiv \omega_{u,v}(x)\pmod{m_{u,v}(x)}$ for $v\in [n_u]$.
   \item Every participant $P_{\theta}=P_{u,v}$ can compute
        \begin{displaymath}
         \begin{aligned}
         c_{\theta}(x)&=(s_{u,v}(x)+\omega_{u,v}(x))\bmod{m_{\theta}(x)}\\
                      &=s'_{q+1}(x)\bmod{m_{\theta}(x)}.
          \end{aligned}
       \end{displaymath}
   \item If $n_u\geq t$ for some $u\in [q]$, then $g_u(x)$ is computable from the public polynomials alone.
   \item If $s'_{q+1}(x)$ is known, then every compartment secret polynomial $s'_u(x)$, and hence the secret $s(x)$, is reconstructed from the public polynomials and $s'_{q+1}(x)$.
\end{enumerate}
\end{lemma}

\begin{proof}
(i) According to Lemma \ref{lemma:1} item (i) and item (iii), we have $\deg(s'_u(x))<T_u$, $\deg(s'_{q+1}(x))<T$, and $T\geq T_u$. Consequently, $\deg(g_u(x))<T$. Recall that $s_{u,v}(x)=s'_u(x)\bmod{m_{u,v}(x)}$ for $v\in [n_u]$, hence
        \begin{displaymath}
         \begin{aligned}
              \omega_{u,v}(x)&=(s'_{q+1}(x)-s_{u,v}(x))\bmod{m_{u,v}(x)}\\
              &=(s'_{q+1}(x)-s'_u(x))\bmod{m_{u,v}(x)}\\
              &=g_u(x)\bmod{m_{u,v}(x)}.
         \end{aligned}
        \end{displaymath}
Namely, $g_u(x)\equiv \omega_{u,v}(x)\pmod{m_{u,v}(x)}$ for $v\in [n_u]$.

(ii) Since $s_{u,v}(x)=s'_u(x)\bmod{m_{u,v}(x)}$ for $v\in [n_u]$, then by the above item (i) we have
        \begin{displaymath}
         \begin{aligned}
         c_{\theta}(x)&=(s_{u,v}(x)+\omega_{u,v}(x))\bmod{m_{\theta}(x)}\\
                      &=(s^{\prime}_{u}(x)+g_u(x))\bmod{m_{\theta}(x)}\\
                      &=s'_{q+1}(x)\bmod{m_{\theta}(x)}.
          \end{aligned}
        \end{displaymath}

(iii) If $n_u\geq t$ for some $u\in [q]$, then $D_{\gamma_u}\geq \sum_{\theta=1}^{t}d_{\theta}=T$ by Lemma \ref{lemma:1} item (ii). According to the above item (i), we have $\deg(g_u(x))<T\leq D_{\gamma_u}$ and $g_u(x)\equiv \omega_{u,v}(x)\pmod{m_{u,v}(x)}$ for $v\in [n_u]$. Therefore, $g_u(x)$ is computable from the public polynomials alone by Lemma \ref{le: CRT}.

(iv) From the above item (i) we have $g_u(x)\equiv \omega_{u,v}(x)\pmod{m_{u,v}(x)}$ for $v\in [n_u]$. This shows that
               \[s'_u(x)\equiv (s'_{q+1}(x)-\omega_{u,v}(x)) \pmod{m_{u,v}(x)} ~\mathrm{for}~v\in [n_u].\]
By Lemma \ref{lemma:1} item (i) and item (iv), it holds that $\deg(s'_u(x))<T_u\leq D_{\gamma_u}$. Therefore, compartment secret polynomials $s'_u(x), u\in [q]$ are recovered by Lemma \ref{le: CRT} if $s'_{q+1}(x)$ is known. Consequently, the secret $s(x)$ is reconstructed by Lemma \ref{lemma:1} item (v).
\end{proof}

Lemma \ref{le: leakage} item (iv) shows that the $q+1$ polynomials $s'_1(x),\ldots,s'_{q+1}(x)$ are not independent, and the secret is a public function of the polynomial $s'_{q+1}(x)$.

\begin{theorem}\label{Th: attack}
In the scheme of \cite{YangCSS2026}, the participants in $\mathcal{B}$ can reconstruct the secret $s(x)$ if one of the following conditions holds, where $\mathcal{B}\subseteq [n]$ and $\mathcal{B}^{(u)}=\mathcal{B}\cap \gamma_u$ for $u\in [q]$.
\begin{enumerate}[label=(\roman*)]
   \item $|\mathcal{B}|\geq t$;
   \item $|\mathcal{B}^{(u)}|\geq t_u$ and $|\mathcal{B}\setminus \gamma_u|+n_u\geq t$ for some $u\in [q]$;
   \item $n_u\geq t$ and $|\mathcal{B}|\geq t_u$ for some $u\in [q]$.
\end{enumerate}
Note that every $\mathcal{B}\in \Gamma_L$ satisfies the conditions (i) and (ii), whereas each of the three conditions may also be satisfied by some unauthorized subset $\mathcal{B}\notin \Gamma_L$ (see Corollary \ref{Cor: witnesses}). Such an unauthorized subset can reconstruct the secret, although Yang, Zhu, Fu and Xia \cite{YangCSS2026} claim that no unauthorized subset obtains any information about the secret.
\end{theorem}

\begin{proof}
Recall that $g_u(x)=s'_{q+1}(x)-s'_u(x)$ and let $c_{\theta}(x)$ be as in Lemma \ref{le: leakage} item (ii). According to Lemma \ref{le: leakage} item (iv),  it suffices to prove that the global secret polynomial $s^{\prime}_{q+1}(x)$ can be recovered from any one of the conditions.

\emph{Case} (i). By Lemma \ref{le: leakage} item (ii), the participants in $\mathcal{B}$ establish the system of congruences
\[s'_{q+1}(x)\equiv c_{\theta}(x) \pmod{m_{\theta}(x)}, \theta\in \mathcal{B}.\]
Since $|\mathcal{B}|\geq t$, from Lemma \ref{lemma:1} item (i) and item (ii), we have
                 \[\deg(s'_{q+1}(x))<T=\sum\limits_{\theta=1}^{t}d_{\theta}\leq D_{\mathcal{B}}.\]
Therefore, $s'_{q+1}(x)$ is recovered by Lemma \ref{le: CRT}.

\emph{Case} (ii). For a given $u\in [q]$, the participants in $\mathcal{B}^{(u)}$ establish the system of congruences
          \[s_{u,v}(x)=s'_u(x)\bmod{m_{u,v}(x)}~\mathrm{for}~P_{u,v}\in \mathcal{B}^{(u)}.\]
Since $|\mathcal{B}^{(u)}|\geq t_u$, from Lemma \ref{lemma:1} item (i) and item (ii), we have
               \[\deg(s'_{u}(x))<T_u=\sum\limits_{\theta=1}^{t_u}d_{\theta}\leq D_{\mathcal{B}^{(u)}}.\]
Therefore, $s'_{u}(x)$ is recovered by Lemma \ref{le: CRT}. Based on this, the set $\mathcal{B}$ can determine $g_u(x)\equiv \omega_{u,v}(x) \!\!\pmod{m_{u,v}(x)}~(v\in [n_u])$ and $g_u(x)\equiv (c_{\theta}(x)-s'_u(x)) \pmod{m_{\theta}(x)}~(\theta\in \mathcal{B}\setminus \gamma_u)$,
where the first family of congruences is publicly known by Lemma \ref{le: leakage} item (i), and the second family uses $c_{\theta}(x), \theta\in \mathcal{B}\setminus \gamma_u$ in Lemma \ref{le: leakage} item (ii) and the recovered $s'_u(x)$. Since $\gamma_u$ and $\mathcal{B}\setminus \gamma_u$ are disjoint and $|\mathcal{B}\setminus \gamma_u|+n_u\geq t$, these $n_u+|\mathcal{B}\setminus \gamma_u|$ pairwise coprime moduli have total degree, by Lemma \ref{lemma:1} item (ii),
                      \[D_{\gamma_u}+D_{\mathcal{B}\setminus \gamma_u}\geq \sum\limits_{\theta=1}^{t}d_{\theta}=T>\deg(g_u(x)).\]
Consequently, the polynomial $g_u(x)$ is recovered by Lemma \ref{le: CRT}, and hence $s'_{q+1}(x)=s'_u(x)+g_u(x)$.

\emph{Case} (iii). Since $n_u\geq t$ for some $u\in [q]$, the polynomial $g_u(x)$ is publicly known by Lemma \ref{le: leakage} item (iii). Recall that $g_u(x)=s'_{q+1}(x)-s'_u(x)$. Hence every $\theta\in \mathcal{B}$, \emph{of whatever compartment}, produces a residue of $s'_u(x)$ by Lemma \ref{le: leakage} item (ii), namely
      \begin{displaymath}
         \begin{aligned}
              s'_u(x)&\equiv (s^{\prime}_{q+1}(x)-g_u(x)) \pmod{m_{\theta}(x)}\\
                     &\equiv (c_{\theta}(x)-g_u(x)) \pmod{m_{\theta}(x)}, \theta\in \mathcal{B}.\\
         \end{aligned}
        \end{displaymath}
Since $|\mathcal{B}|\geq t_u$, we have $\deg(s'_u(x))<T_u\leq D_{\mathcal{B}}$ by Lemma \ref{lemma:1} item (i) and (ii). Therefore, the compartment secret polynomial $s'_u(x)$ is recovered, and hence $s'_{q+1}(x)=s'_u(x)+g_u(x)$.

Finally, let $\mathcal{B}\in \Gamma_L$. Then $|\mathcal{B}|\geq t$, so the condition (i) holds. Moreover, for every $u\in [q]$, $|\mathcal{B}^{(u)}|\geq t_u$ together with
         \[|\mathcal{B}\setminus \gamma_u|+n_u\geq |\mathcal{B}\setminus \gamma_u|+|\mathcal{B}^{(u)}|=|\mathcal{B}|\geq t\]
shows that the condition (ii) holds as well.
\end{proof}

\begin{corollary}\label{Cor: witnesses}
Each of the three conditions of Theorem \ref{Th: attack} is satisfied by some unauthorized subset, in the following sense.
\begin{enumerate}[label=(\roman*)]
  \item If $n-n_u+t_u-1\geq t$ for some $u\in [q]$, then the subset consisting of all the participants outside $\gamma_u$ together with $t_u-1$ participants of $\gamma_u$ is unauthorized and satisfies Theorem \ref{Th: attack} item (i).
  \item If $t_u<n_u$ for some $u\in [q]$, then the subset consisting of $t_u$ participants of $\gamma_u$ together with $\max\{0,t-n_u\}$ participants outside $\gamma_u$ is unauthorized and satisfies Theorem \ref{Th: attack} item (ii).
  \item If $n_u\geq t$ for some $u\in [q]$, then any subset consisting of $t_u$ participants of $\gamma_u$ is unauthorized and satisfies Theorem \ref{Th: attack} item (iii).
\end{enumerate}
Consequently, the scheme of \cite{YangCSS2026} is insecure unless $t_u=n_u$ for all $u\in [q]$, in which case $t=n$ and $\Gamma_L=\{\mathcal{P}\}$.
\end{corollary}

\begin{proof}
Denote the subsets constructed in items (i)-(iii) be $\mathcal{B}_1, \mathcal{B}_2$ and $\mathcal{B}_3$, respectively.

(i) Since $|\mathcal{B}_1|=(n-n_u)+(t_u-1)\geq t$, Theorem \ref{Th: attack} item (i) holds, while $|\mathcal{B}_1\cap \gamma_u|=t_u-1<t_u$ shows that $\mathcal{B}_1\notin \Gamma_L$.

(ii) Since $t\leq n$, we have $0\leq\max\{0,t-n_u\}\leq n-n_u$, which ensures that the subset $\mathcal{B}_2$ is well defined. Since $|\mathcal{B}_2\cap \gamma_u|=|\mathcal{B}_2^{(u)}|=t_u$ and $|\mathcal{B}_2\setminus \gamma_u|+n_u=\max\{0,t-n_u\}+n_u\geq t$, Theorem \ref{Th: attack} item (ii) holds. Now we show that $\mathcal{B}_2\notin \Gamma_L$. If $n_u\geq t$, then $\mathcal{B}_2\subseteq \gamma_u$, and $|\mathcal{B}_2\cap \gamma_w|=0<t_w$ for any $w\neq u$ (such a $w$ exists since $q\geq 2$). If $n_u<t$, then $|\mathcal{B}_2|=t_u+t-n_u<t$ because $t_u<n_u$. In both cases $\mathcal{B}_2\notin \Gamma_L$.

(iii) Since $n_u\geq t$ for some $u\in [q]$ and $|\mathcal{B}_3|=t_u$, Theorem \ref{Th: attack} item (iii) holds. Moreover, $\mathcal{B}_3\subseteq \gamma_u$, which implies that $|\mathcal{B}_3\cap \gamma_w|=0<t_w$ for any $w\neq u$ (such a $w$ exists since $q\geq 2$), so $\mathcal{B}_3\notin \Gamma_L$.

Finally, assume that the scheme is secure. By the above item (ii), $t_u=n_u$ for all $u\in [q]$. Hence $\sum_{u=1}^{q}t_u=\sum_{u=1}^{q}n_u=n$, and $\sum_{u=1}^{q}t_u\leq t\leq n$ implies $t=n$. Since every authorized subset $\mathcal{A}$ satisfies $|\mathcal{A}\cap \gamma_u|\geq t_u=n_u$ for all $u\in [q]$, it follows that $\mathcal{A}\supseteq \gamma_u$ for all $u\in [q]$. Therefore $\Gamma_L=\{\mathcal{P}\}$.
\end{proof}

\begin{remark}\label{Re: degeneration}
Theorem \ref{Th: attack} item (i) shows that the compartment thresholds play no role at all: any $t$ participants reconstruct the secret, so the family of sets that can reconstruct the secret contains $\{\mathcal{B}\subseteq \mathcal{P}: |\mathcal{B}|\geq t\}$, and the scheme realises at most a $(t,n)$-threshold access structure. Theorem \ref{Th: attack} item (iii) is more damaging still: as soon as one compartment satisfies $n_u\geq t$, any $t_u$ participants reconstruct the secret, and the scheme realises at most the $(t_u,n)$-threshold access structure. Since $t_u$ may be as small as $1$, the scheme may be broken by a single participant.
\end{remark}

\subsection{The flaw in the security analysis}\label{Subsec: discussion}

In the perfectness argument of \cite{YangCSS2026} (Lemmas III.1, III.2 and Theorem III.3 therein), for an unauthorized subset $\mathcal{A}_1$ with $|\mathcal{A}_1\cap \gamma_i|<t_i$ for some $i\in [q]$, or with $|\mathcal{A}_1|<t$, the following sets are introduced:
      \begin{displaymath}
        \begin{aligned}
        \mathcal{G}_i=&\{g(x): \deg(g(x))<T_i,\ g(x)\equiv c_{\theta^{(i)}}(x) \pmod{m_{\theta^{(i)}}(x)}\\
                     &\mathrm{for}~P_{\theta^{(i)}}\in \mathcal{A}_1^{(i)}=\mathcal{A}_1\cap \gamma_i ~\mathrm{of}~|\mathcal{A}_1^{(i)}|<t_i\}
        \end{aligned}
      \end{displaymath}
and
     \begin{displaymath}
        \begin{aligned}
        \mathcal{G}=&\{g(x): \deg(g(x))<T,\ g(x)\equiv c_{\theta}(x) \pmod{m_{\theta}(x)}\\
                     &\mathrm{for}~P_{\theta}\in \mathcal{A}_1~\mathrm{of}~|\mathcal{A}_1|<t\}.
        \end{aligned}
     \end{displaymath}
The participants in $\mathcal{A}_1$ choose a polynomial $\widetilde{s_i^{\prime}}(x)$ randomly from $\mathcal{G}_i$ for some $i\in [q]$, or a polynomial $\widetilde{s^{\prime}}(x)$ randomly from $\mathcal{G}$, and then compute $\widetilde{s_i^{\prime}}(x)\bmod{m_0(x)}$ or $\widetilde{s^{\prime}}(x)\bmod{m_0(x)}$ to obtain the secret. Yang et al. conclude that, from the adversary's point of view, the secret is uniformly distributed on $\mathcal{S}$.

Note that the congruences defining $\mathcal{G}_i$ involve \emph{only the shares of $\mathcal{A}_1$}; the published polynomials $\omega_{\theta}(x)$, $\theta\in [n]$, do not appear.
The congruences defining $\mathcal{G}$ involve \emph{only the shares of $\mathcal{A}_1$ and the partial published polynomials $\omega_{\theta}(x), \theta\in \mathcal{A}_1$}. However, by any reasonable definition of perfectness for a scheme with public data, the published polynomials are part of the adversary's view. Therefore, the fact that the sets $\mathcal{G}_i$ and $\mathcal{G}$ do not involve all the published polynomials is the flaw in the security analysis of \cite{YangCSS2026}.

\begin{remark}\label{Re:2}
Similar flaws also appear in the hierarchical secret sharing schemes of the same first author \cite{Yangjing2024,YangjingHSS2026}, whose security analysis omits the published polynomials in exactly the same way.
\end{remark}

\section{Conclusion}\label{Sec: conclusion}

We have shown that the compartmented secret sharing scheme of Yang, Zhu, Fu and Xia \cite{YangCSS2026} is not perfect, and therefore not ideal. The scheme is insecure for every choice of the parameters except the trivial one in which the only authorized subset is $\mathcal{P}$ itself. Similar flaws also appear in the hierarchical secret sharing schemes of the same first author \cite{Yangjing2024,YangjingHSS2026}. It therefore remains an open problem to design an ideal CRT-based CSSS for LCASs and an ideal CRT-based hierarchical secret sharing scheme.

\begin{IEEEbiographynophoto}{Jian Ding}
received the MS degree in applied mathematics from Hefei University of Technology, China in 2010, and the PhD degree in Cyberspace Security from Fujian Normal University, China in 2022. He is currently a Professor at School of Mathematics and Big Data, Chaohu University, China. His research interests are in cryptography and network security, especially in secret sharing and its applications.
\end{IEEEbiographynophoto}

\begin{IEEEbiographynophoto}{Zhiyuan Guo}
 is currently a graduate student at School of Mathematics and Statistics, Hefei University, China. His research interests are in cryptography and network security, especially in secret sharing and its applications.
\end{IEEEbiographynophoto}

\begin{IEEEbiographynophoto}{Hongju Li}
received the MS degree in statistics from Guilin University of Technology, China in 2010. She is currently a Associate Professor at School of Mathematics and Big Data, Chaohu University, China. Her research interests are in algebraic coding and cryptography.
\end{IEEEbiographynophoto}

\begin{IEEEbiographynophoto}{Yuying Man}
received the MS degree in applied mathematics from South Central Minzu University, China in 2022, and the PhD degree in Mathematics from Hubei University, China in 2025. She is currently an Assistant Lecturer at School of Mathematics and Big Data, Chaohu University, China. Her research interests are in cryptography, especially in cryptographic functions.
\end{IEEEbiographynophoto}

\end{document}